\documentclass[11pt]{article}

\usepackage[letterpaper, left=1in, right=1in, bottom=1in, top=1in]{geometry}
\usepackage{authblk}
\usepackage[dvipdfmx]{graphicx}
\usepackage{xcolor}
\usepackage{fancyhdr}
\usepackage{times}
\usepackage{url}
\usepackage{pifont}
\usepackage{multirow}
\usepackage{booktabs}
\usepackage{amsmath,amssymb,amsthm}
\usepackage{cases}
\usepackage{tikz}
\usepackage{algorithm}
\usepackage{wrapfig}
\usepackage{afterpage}
\usepackage{enumitem}
\usepackage{mathtools}
\usepackage{bm}
\usepackage{comment}

\newtheorem{theorem}{Theorem}
\newtheorem{corollary}{Corollary}
\newtheorem{lemma}{Lemma}

\newtheorem{example}{Example}

\newcommand{\STree}{\mathsf{STree}}
\newcommand{\occ}{\mathit{occ}}
\newcommand{\leaf}[1]{\mathit{leaf}(#1)}
\newcommand{\mtt}{\mathtt}
\newcommand{\plp}[1]{\mathit{PLP}(#1)}
\newcommand{\credit}{\mathit{credit}}

\begin{document}

\title{Sliding suffix trees revisited}
\author[1]{Laurentius~Leonard}
\author[1]{Shunsuke~Inenaga}
\author[2]{Hideo~Bannai}
\author[3]{Takuya~Mieno}

\affil[1]{Kyushu University, Japan}
\affil[3]{Institute of Science Tokyo, Japan}
\affil[4]{University of Electro-Communications, Japan}

\date{}
\maketitle

\begin{abstract}
  The \emph{sliding suffix tree} (Fiala \& Greene, 1989) is a suffix tree
  that is maintained for a sliding window $W_i = T[i..i+d-1]$ of size $d$
  that shifts over an input text $T$ of length $n$ from left to right, for increasing $i = 1, \ldots, n-d+1$.
  It is known that the sliding suffix tree can be maintained in $O(n \log \sigma)$ time with $O(d)$ space, where $\sigma$ is the alphabet size.
  Updating the sliding suffix tree from $W_i = T[i..i+d-1]$ to $W_{i+1} = T[i+1..i+d]$ requires the following three major tasks:
  (1) Delete the leaf that represents the longest suffix $W_i$,
  (2) Insert new leaves that represent the suffixes of $W_{i+1}$ that appear exactly once in $W_{i+1}$, and
  (3) After the leaf deletion due to Task (1) and each leaf insertion due to Task (2), maintain the label $\langle \ell, r \rangle$ of every edge as a \emph{valid} pair in the new window $W_{i+1}$, such that $i+1 \leq \ell \leq r \leq i+d$.
  In this paper, we present the first algorithm that performs Task (3) in $O(1)$ worst-case time per node deletion/insertion,
  which leads to another alternative to efficient sliding suffix tree construction.
  This is an improvement over
  the existing algorithms by Larsson (1996, 1999) and by Senft (2005)
  both of which can only perform Task (3) in $O(1)$ \emph{amortized} time.
  Our key data structure is a non-trivial extension of \emph{leaf pointers},
  which were originally proposed by Brodnik and Jekovec (2018) for pattern matching with sliding suffix trees.
\end{abstract}

\section{Introduction}

The suffix tree~\cite{Weiner} of a string $T$, denoted $\STree(T)$,
is an edge-labeled rooted tree that represents all suffixes of $T$.
It is well known that
$\STree(T)$ can be stored in $O(n)$ space for any string $T$ of length $n$
by representing each edge label $x$ with a pair $\langle i,j \rangle$ of positions in $T$
such that $x = T[i..j]$.
A suffix tree is a powerful string indexing data structure with a myriad of applications~\cite{DanGusfield}.

Among the variations of the suffix tree is the sliding suffix tree, proposed by
Fiala and Greene~\cite{Fiala1989}.
Let $W_i = T[i..i+d-1]$ be a \emph{sliding window} of (fixed) size $d$ over a text $T$ of length $n$.
The sliding suffix tree is the suffix tree of sliding window $W_i$ that is maintained and updated for increasing $i = 1, \ldots, n-d+1$.
Sliding suffix trees are core data structures
in data compression~\cite{Fiala1989,Larsson1996,Larsson1999} including the original variant of the Lempel-Ziv 77 compression~\cite{LZ77},
as well as pattern matching~\cite{Brodnik2018},
computing maximal absent words (MAWs)~\cite{CrochemoreHKMPR20},
minimal unique substrings (MUSs)~\cite{MienoFNIBT22},
and string net frequencies (NFs)~\cite{MienoI25},
in the sliding window model.

Shifting the current window $W_i = T[i..i+d-1]$ to the next window $W_{i+1} = T[i+1..i+d]$ involves deleting the leftmost character $T[i]$ and appending a new character $T[i+d]$ to the window.
Thus, updating $\STree(W_i)$ to $\STree(W_{i+1})$ requires the following major tasks in the data structure:
\begin{enumerate}
  \item[(1)] Delete the leaf that represents the longest suffix $W_i = T[i..i+d-1]$.
  \item[(2)] Insert new leaves that represent the suffixes of $W_{i+1} = T[i+1..i+d]$ that appear exactly once in $W_{i+1}$, and
  \item[(3)] After the leaf deletion due to Task (1) and each leaf insertion due to Task (2), maintain the label $\langle \ell, r \rangle$ of every edge as a \emph{valid} pair in the new window $W_{i+1} = T[i+1..i+d]$, such that $i+1 \leq \ell \leq r \leq i+d$.
\end{enumerate}
We do not scan all edges after each shift. Rather, whenever a local topological update creates, modifies, or needs to materialize an edge label, we compute a fresh index-pair for that edge in $O(1)$ worst-case time.
While Task (2) also appears in online suffix tree constructions~\cite{Ukkonen}, Tasks (1) and (3) are specific to the sliding window model.

Larsson~\cite{Larsson1996,Larsson1999} pointed out that
Task (1) is a mere operation that takes $O(1)$ \emph{worst-case} time.
Larsson~\cite{Larsson1996,Larsson1999} also showed that
Task (2) can be done in $O(1)$ \emph{amortized} time per character,
by adapting Ukkonen's online suffix tree construction~\cite{Ukkonen}.

The focus of this paper is Task (3) of updating edge labels,
for which the following two existing approaches are known:
\begin{itemize}
  \item The first one is the so-called \emph{credit-based} approach
    proposed by Larsson~\cite{Larsson1996,Larsson1999},
    that is based on the ideas from Fiala and Greene's earlier work~\cite{Fiala1989}.
    The basic idea of this method is to propagate a credit to each edge label that needs to be updated, and it can be shown that the time cost for updating each edge label $\langle \ell, r \rangle$ takes $O(1)$ \emph{amortized} time.
    The worst-case time is however $O(d)$ per left-end character deletion and right-end character insertion.
    In addition, in Section~\ref{sec:loweround} we show that the worst-case time of the credit-based method is $\Omega(d)$ for a single leaf insertion or deletion on some string.
  \item The second one is the so-called \emph{batch-based} approach proposed by Senft~\cite{Senft2005}. Senft's method uses standard batch-updates on overlapping blocks of size $2d$ over the text $T$,
    shifting a window of size $d$ in each block to the right,
    reconstructing the suffix tree from scratch when the window no longer fits in the block, moving on to the next block prefixed by the window.
    It is clear that the batch-based approach takes $O(1)$ \emph{amortized} time.
    However,
    the worst-case time is $O(d)$ per character deletion/insertion as well as for leaf deletion/insertion.
\end{itemize}

In this paper, we present the first method, named the \emph{leaf-pointer} based approach, that is able to perform Task (3) in $O(1)$ \emph{worst-case} time for each edge label update that is incidental to a leaf insertion or deletion.
We first make a simple observation that maintaining the label for each edge $u \rightarrow v$ reduces to maintaining a pointer from node $v$ to a leaf in the subtree rooted at $v$ (Lemma~\ref{the:getPair} in Section~\ref{sec:fresh_index}).
Then, we present how to maintain a pointer from each node to
a specific leaf in its subtree in $O(1)$ worst-case time (Section~\ref{sec:update_leafpointers}).
This is achieved by introducing our concepts of \emph{primary} and \emph{secondary} nodes and by maintaining leaf pointers built on them, which may be of independent interest and may be used in other applications.
Table~\ref{table:comparisons} summarizes the comparison of the existing methods and our proposed method.
\begin{table}[hbt]
  \centering
  \caption{Comparison of algorithms for maintaining a suffix tree of a sliding window of size $d$ on a string of length $n$ over an alphabet of size $\sigma$. In the table, ``amortized'' and ``worst case'' refer to ``amortized $O(1)$ time'' and ``worst case $O(1)$ time'', respectively. The $O(\log \sigma)$ term in the total time complexities is due to maintaining branches of nodes.}
  \begin{tabular}{|l||c|c|c|c|c|}
    \hline
    algorithms                             & Task (1)   & Task (2)           & Task (3) & total time & working space \\
                                           &            &                    & per node deletion/insertion &  &  \\ \hline \hline
    credits~\cite{Larsson1996,Larsson1999} & worst case & amortized
                                           & amortized  & $O(n \log \sigma)$ & $O(d)$                                \\ \hline
    batch~\cite{Senft2005}                 & worst case & amortized
                                           & amortized  & $O(n \log \sigma)$ & $O(d)$                                \\ \hline
    leaf pointers [ours]                   & worst case & amortized
                                           & worst case & $O(n \log \sigma)$ & $O(d)$                                \\ \hline
  \end{tabular}
  \label{table:comparisons}
\end{table}

\paragraph*{Related work.}
Brodnik and Jekovec~\cite{Brodnik2018} used a similar notion to our leaf pointers. They showed how to perform pattern matching in $O(m \log \sigma + \occ)$ time on the sliding suffix tree with $O(d)$ working space, where $m$ is the query pattern length and $\occ$ is the number of \emph{all} pattern occurrences in the current window.
Since Brodnik and Jekovec~\cite{Brodnik2018} utilize the credit-based algorithm to maintain their leaf pointers, their method takes $O(1)$ \emph{amortized} time for each leaf insertion/deletion.

\section{Preliminaries}
\subsection{Strings}
\label{subsecPrelim}
Let $\Sigma$ be an alphabet of size $\sigma$. An element of the set $\Sigma^*$ is a string.
The length of a string $w$ is denoted by $|w|$.
For a string $w = pts$, $p$, $t$, and $s$ are called a
\emph{prefix}, \emph{substring}, and \emph{suffix} of $w$, respectively.
A suffix $s$ of $w$ that occurs twice or more in $w$ is said to be \emph{repeating} in $w$.
For a string $w$,
$w[i]$ denotes the $i$-th character of string $w$ for $1 \leq i \leq |w|$, and
$w[i..j]$ denotes the substring $w[i]w[i+1]\cdots w[j]$ of string $w$ for $1 \leq i \leq j \leq |w|$.

\subsection{Sliding suffix trees}
The suffix tree~\cite{Weiner,DanGusfield} of a string $T$, denoted $\STree(T)$,
is the compacted trie of all suffixes of $T$.
Note that each non-leaf node of $\STree(T)$ has at least two children (unless $T$ is a unary string),
and thus, the number of nodes in $\STree(T)$ is linear in the length of $T$.
Each edge is labeled by a non-empty substring of $T$ and the labels of all outgoing edges of the same node begin with distinct characters. Any substring $w$ of $T$ can be associated with a corresponding location (either on a node or a position on an edge) in the suffix tree where the concatenated path label from the root to the location is $w$.
Depending on the context, we will use strings and their locations in the suffix tree interchangeably.
An edge of a suffix tree from node $u$ to node $v$ is denoted by $u \to v$.
The string label $x$ for each edge $u \to v$ is represented by an index-pair $\langle\ell, r\rangle$ such that $x = T[\ell .. r]$.
This way, $\STree(T)$ can be stored in space linear in the length of $T$.
The sliding suffix tree~\cite{Fiala1989,Senft2005,Brodnik2018}
is the suffix tree for a sliding window $W_i = T[i..i+d-1]$ of size $d$,
that shifts over the input string $T$ of length $n$ from left to right,
for increasing $i = 1, \ldots, n-d+1$.
This is a version of suffix trees a.k.a. Ukkonen trees~\cite{Ukkonen}, in which
non-repeating suffixes of $W_i$ are represented by leaves
while repeating suffixes of $W_i$ are represented by (possibly implicit, non-branching) internal nodes.
Let us consider a sliding window $W_i = T[i..i+d-1]$ over an input string $T$.
For ease of description, let $\ell_i = i$ and $r_i = i+d-1$.
Let $W_i[s..r_i]$ be the longest repeating suffix of $W_i$.
For any $\ell_i \leq k < s$, $\mathit{leaf}_i(k)$ denotes the leaf of
the suffix tree of $W_i$ that represents the suffix $W_i[k..r_i]$.
The sliding suffix tree $\STree(W_i)$ for $W_i$
consists of the leaves $\mathit{leaf}_i(k)$ for $\ell_i \leq k < s$,
branching internal nodes, and the root.
When there is no risk of confusion, we simply denote it by $\leaf{k}$.
Figure~\ref{fig:slidingsuftree} shows an example of how a sliding suffix tree changes between two iterations.
On the left side of Figure~\ref{fig:slidingsuftree}, strings $\mtt{a}, \mtt{abaca}$ and the empty string are considered to be located on nodes, while strings $\mtt{bac}, \mtt{c}$ are considered to be located on edges.

\begin{figure}[tb]
  \centering
  \includegraphics[scale=0.4]{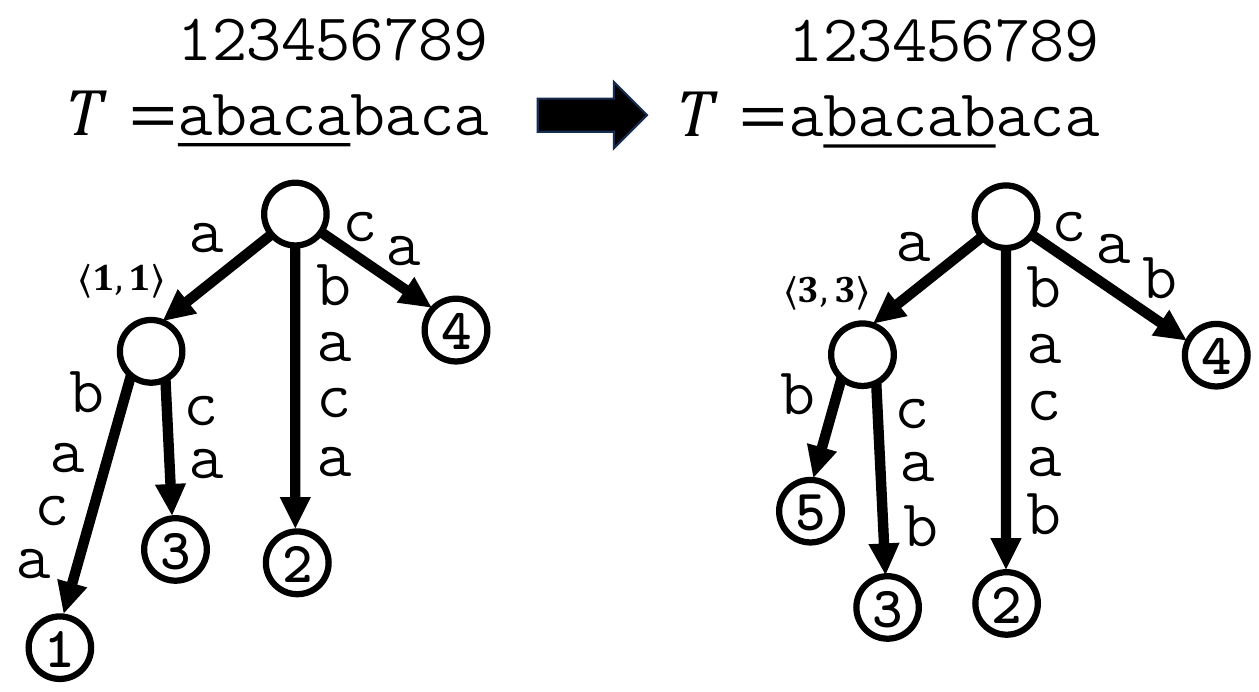}
  \caption{Sliding suffix tree across two iterations with $T=\mtt{abacabaca}$ and $|W|=5$. $\leaf{1}$ is deleted in the second iteration due to its corresponding suffix $\mtt{abaca}$, which would grow into $\mtt{abacab}$, not existing in the new window. The edge labeled $\mtt{a}$ shows an example of an edge whose index-pair becomes outdated; if it was represented by $\langle 1,1 \rangle$ in the
  first iteration, an update is required as the index $1$ contained in the interval is no longer inside the window.}
  \label{fig:slidingsuftree}
\end{figure}

In our sliding suffix tree algorithm,
only the last $d$ characters of $T$
are explicitly stored to achieve $O(d)$ space efficiency, as done in the credit-based approach by Larsson~\cite{Larsson1996,Larsson1999}.
Thus, any edge whose index-pair contains an index $k<\ell_i$, or in other words any edge with an index-pair corresponding to an occurrence of the edge-label that is no longer entirely inside the window, needs to be updated
to a \emph{fresh} index-pair $\langle \ell, r \rangle$ such that $[\ell..r] \subseteq [\ell_i..r_i]$.
A {\em leaf pointer} on a suffix tree is defined as a pointer from each internal node to any of its descendant leaves.
Leaf pointers were originally proposed by Brodnik and Jekovec~\cite{Brodnik2018} for their use in online pattern matching.
Figure~\ref{fig:LeafPointers} shows configuration examples of leaf pointers.
Unlike with online suffix trees, maintaining them for a sliding suffix tree is not a trivial task.
Brodnik and Jekovec~\cite{Brodnik2018} claimed that
one can update each leaf pointer in $O(1)$ amortized time per iteration,
by extending the credit-based method by Larsson~\cite{Larsson1996,Larsson1999}.
\begin{figure}[tb]
  \centering
  \centering
  \includegraphics[scale=0.4]{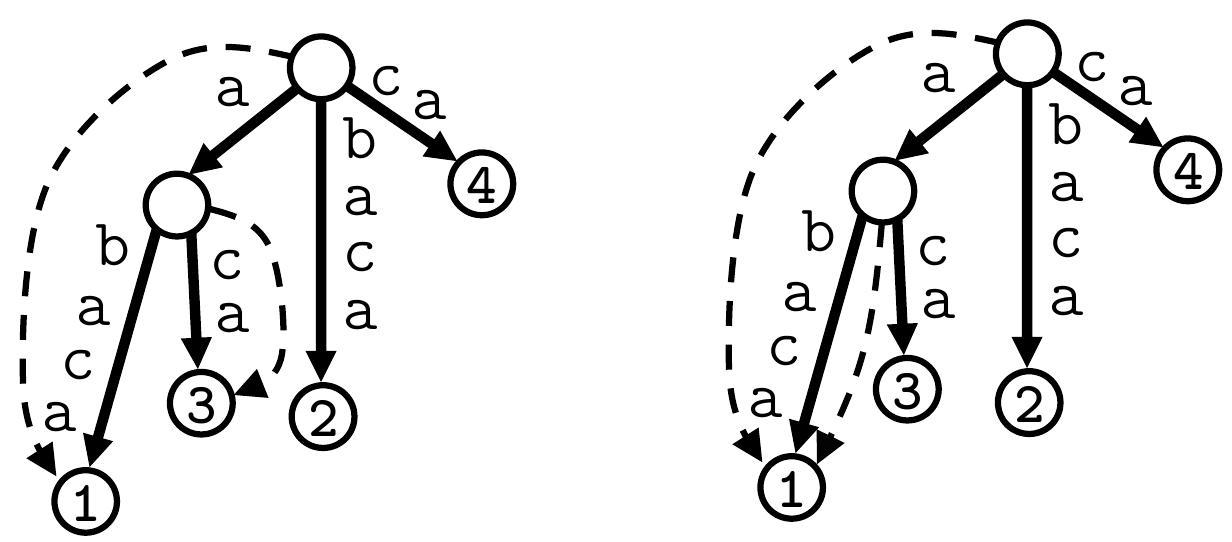}
  \caption{Two possible configurations of leaf pointers for the suffix tree of $\mtt{abaca}$. The dashed arrows depict leaf pointers.}
  \label{fig:LeafPointers}
\end{figure}

In the rest of this paper, we will prove the following theorem:
\begin{theorem} \label{theo:main}
  Let $T$ be a string of length $n$
  and $d$ be the size of the sliding window $W_i = T[\ell_i..r_i]$ over $T$.
  There is an algorithm of $O(d)$ space that maintains
  a leaf pointer of every node in the sliding suffix tree
  in $O(1)$ worst-case time per update.
\end{theorem}
\section{Getting fresh index-pairs from leaf pointers}\label{sec:fresh_index}

We assume that the topology of the suffix tree is maintained, along with 
the string length corresponding to each node and the index $k$ for each $\leaf{k}$.
Provided that leaf pointers are maintained,
we can find a fresh index-pair of any edge $u\rightarrow v$ with edge label $x$ in constant time as follows:
If $v$ is a leaf, it is clear that $\langle k_v, r_i \rangle$ is a fresh index-pair of the edge where $v = \leaf{k_v}$ and $r_i$ is the right-end of the current window.
Otherwise,
traverse the leaf pointer of $v=ux$ to the destination leaf $\leaf{k}$.
The existence of the leaf implies that the suffix $T[k .. r_i]$ must occur inside the current window $W_i$,
and $ux$ is a prefix of this suffix.
Thus, $\langle k +|u| , k + |ux|-1 \rangle$ is a fresh index-pair of $u\rightarrow v$,
corresponding to the substring $x$.
Additionally,
the edge label $\langle k +|u| , k + |ux|-1 \rangle$
is also \emph{strongly fresh},
meaning that
the preceding $u$ at start-index $k$ is also inside the window.
Figure~\ref{fig:GetIndexPair1} illustrates the relationship between $u$,$v$, and $\leaf{k}$.
\begin{figure}[tb]
  \centering
  \begin{minipage}{.46\textwidth}
    \centering
    \includegraphics[scale=0.5]{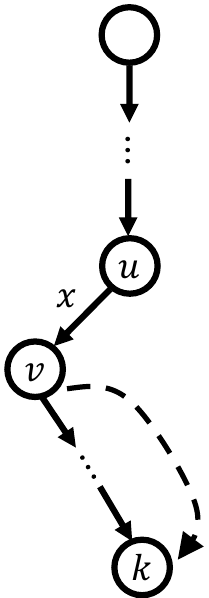}
    \caption{A figure showing the relationship between $u$,$v$, and $\leaf{k}$ in the context of computing the index-pair in the general case.}
    \label{fig:GetIndexPair1}
  \end{minipage}\hspace{\fill}
  \begin{minipage}{.49\textwidth}
    \centering
    \includegraphics[scale=0.5]{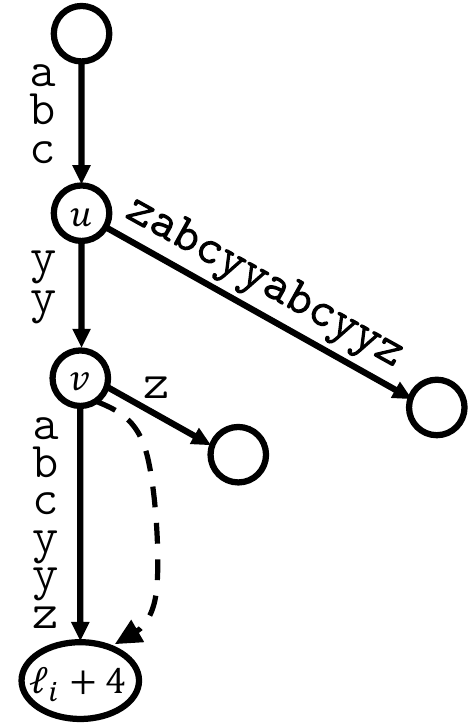}
    \caption{The part of the suffix tree that starts with the character $\mtt{a}$ in Example \ref{ex:GetIndexPair}.}
    \label{fig:GetIndexPair2}
  \end{minipage}
\end{figure}
The following lemma immediately follows:
\begin{lemma}\label{the:getPair}
  Assume that the leaf pointers for all internal nodes are maintained.
  Then, for any edge $u\rightarrow v$, 
  one can obtain in $O(1)$ worst-case time an index-pair $\langle \ell, r\rangle$ such that
  $T[\ell .. r]$ matches the edge's label, and $[\ell - |u| .. r] \subseteq [\ell_i .. r_i]$.
\end{lemma}
\begin{example}
  Consider $W_i=\mtt{abczabcyyabcyyz}$ and ${T=PW_i}$ for some $P\in \Sigma^*$,
  $u=\mtt{abc}$, $v=ux=\mtt{abcyy}$.
  Then, $u$ and $v$ are branching nodes.
  Regardless of the content of $P$, we can get an index-pair of the edge $u\rightarrow v$ with label $x=\mtt{yy}$
  as follows:
  Traverse the leaf pointer of $v$ to arrive at one of its descendant leaves, say,
  $\leaf{\ell_i+4}$ corresponding to the suffix $\mtt{abcyyabcyyz}$.
  Then, we get the index-pair $\langle \ell_i+4+|u|, \ell_i+4+|u|+|x|-1\rangle=\langle \ell_i+7, \ell_i+8 \rangle$,
  which gives us the correct label, since $W_i[8..9]=\mtt{yy}$.
  Figure~\ref{fig:GetIndexPair2} shows the relevant part of the suffix tree of $W_i$, namely the locations starting with $\mtt{a}$.
  \label{ex:GetIndexPair}
\end{example}

We note that, even if the edge $u\rightarrow v$ with index-pair $\langle \ell , r\rangle$ is a fresh edge in the weak sense, i.e.,
$[\ell .. r] \subseteq [\ell_i .. r_i]$,
we cannot use the opposite of the above method to obtain a leaf pointer 
of node $u$.
We can obtain the start-index $\ell - |u|$, but the corresponding $\leaf{\ell-|u|}$ may already be deleted, i.e., $\ell -|u|<\ell_i$,
unless the edge label is strongly fresh, which cannot be assumed for sliding suffix tree construction algorithms in general.

Theorem~\ref{theo:main} and Lemma~\ref{the:getPair} immediately lead to the following:
\begin{corollary}
  Let $T$ be a string of length $n$
  and $d$ be the size of the sliding window $W_i = T[\ell_i..r_i]$ over $T$.
  There is an algorithm of $O(d)$ space that resolves Task (3)
  in $O(1)$ worst-case time per leaf insertion/deletion.
\end{corollary}

We will prove Theorem~\ref{theo:main} in Section~\ref{sec:update_leafpointers}.

\section{Maintaining leaf pointers in $O(1)$ worst-case time}
\label{sec:update_leafpointers}

Our leaf pointer algorithm maintains the following invariants:
For each internal node $x$ with $c>0$ children,
exactly one child node is designated as a \emph{primary} node, while each of the remaining $c-1$ children are designated as \emph{secondary} nodes.
An edge is similarly called primary or secondary depending on its destination node.
For convenience, the root will be considered a secondary node.

Let $S$ be the set of (arbitrarily chosen) primary nodes.
Then, for each secondary internal node $x$, we maintain the \emph{primary leaf pointer} $\mathit{PLP}_S(x)$ that points to the leaf node reached by traversing down the primary edges from $x$.
It is clear that the path of primary edges from a secondary node $x$
leads to a leaf (see also Figure~\ref{fig:leaf_pointer_decomposition} for illustrations).
For each secondary leaf node $x$, we let $\mathit{PLP}_S(x) = x$.
When the choice $S$ of primary nodes in the sliding suffix tree
is clear from the context, we abbreviate $\mathit{PLP}_S(x)$ as $\plp{x}$.

\begin{figure}[h]
  \centering
  \includegraphics[scale=0.4]{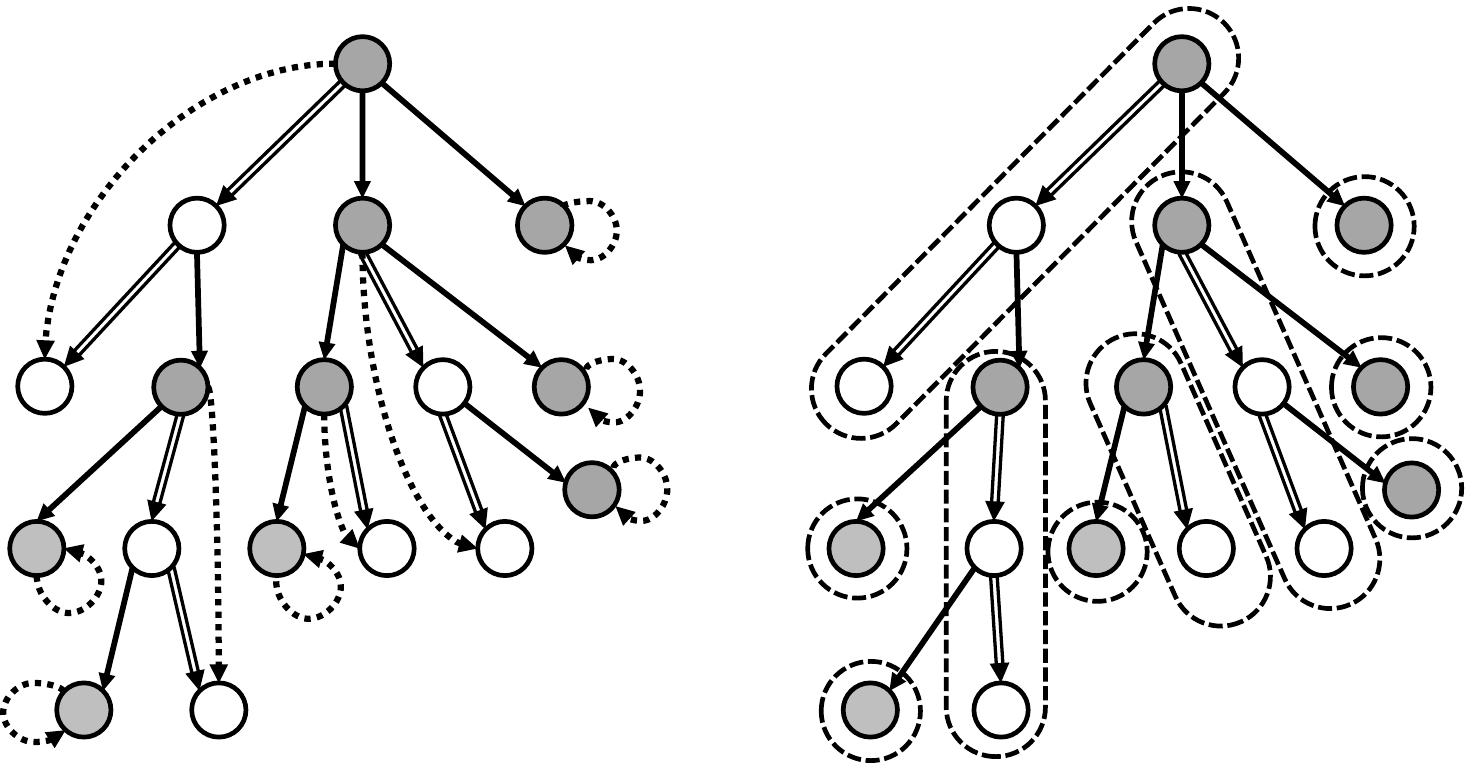}
  \caption{The white and gray nodes represent primary and secondary nodes, respectively. The double-lined solid arrows represent primary edges, while the single-lined solid arrows represent secondary edges. Left: The dotted arrow from each secondary node $x$ represents the primary leaf pointer $\plp{x}$. Right: Our primary/secondary edges induce a path-decomposition of the tree which can be seen as a relaxed version of the well-known heavy-path decomposition, where our primary edge of each node can be \emph{any} of its outgoing edges.}
  \label{fig:leaf_pointer_decomposition}
\end{figure}

Assuming the above invariants are maintained, we can answer the leaf pointer query for any internal node $x$ as follows:
If $x$ is secondary, return $\plp{x}$.
Otherwise, $x$ is a primary (thus non-root) internal node, meaning it has at least two children, so it must have a secondary child $y$.
In this case, return $\plp{y}$ for an arbitrarily chosen secondary child $y$.
See Figure~\ref{fig:plp} for a concrete example.

\begin{figure}[h!]
  \centering
  \includegraphics[scale=0.4]{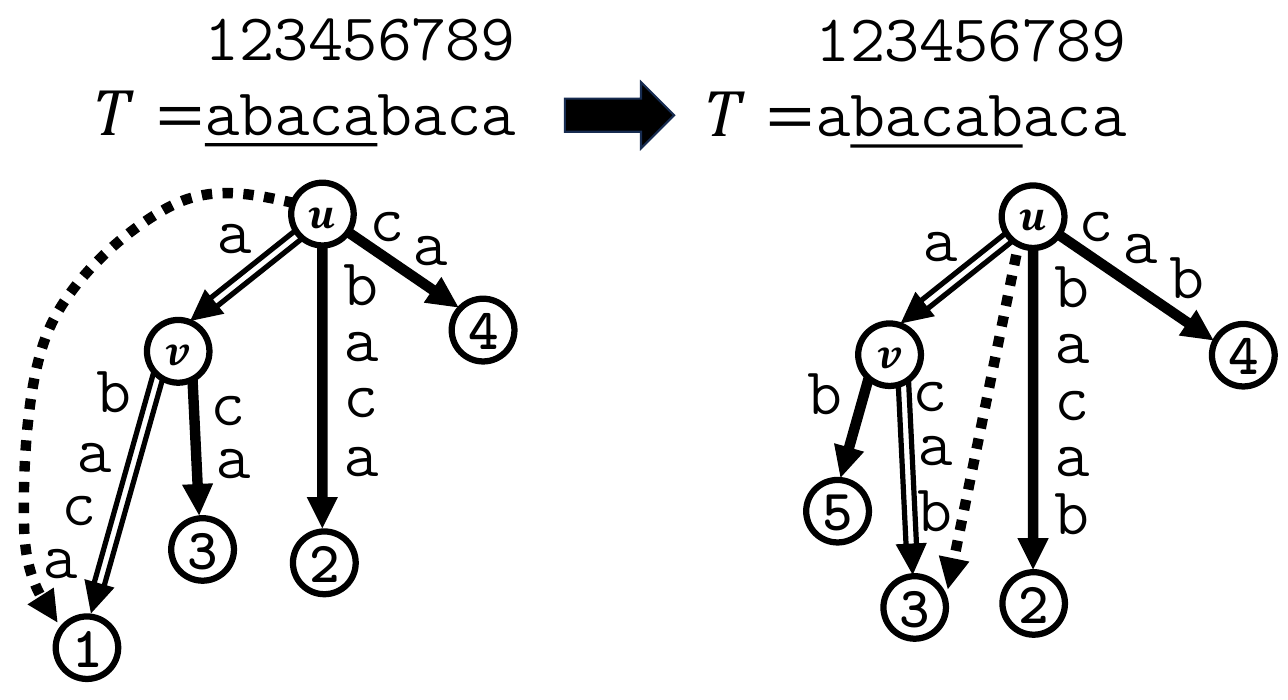}\caption{
    Primary leaf pointers maintained across two iterations of a sliding suffix tree.
    Primary edges are depicted by double-lined arrows, secondary edges by single-lined arrows, and primary leaf pointers by dotted arrows.
    The primary leaf pointers from each secondary leaf that point to themselves are omitted in this figure.
    Left (before the update): The secondary branching node $u$ returns $\plp{u} = \leaf{1}$,
    while a leaf pointer query on the primary branching node $v$ returns the primary leaf pointer of its secondary child $\leaf{3}$, resulting in $\plp{\leaf{3}}=\leaf{3}$.
    Right (after the update): Queries on $u$ and $v$ return $\leaf{3}$ and $\leaf{5}$, respectively.
  }
  \label{fig:plp}
\end{figure}

The invariants above also guarantee that each leaf only has one primary leaf pointer pointing to it, making updates in constant time and linear space feasible.

\begin{lemma} \label{lem:onePLP}
  For each leaf $u$ of the tree, there is exactly one primary leaf pointer that points to $u$.
  Equivalently, the primary paths induce a one-to-one correspondence between secondary nodes and leaves. Hence every leaf has a unique inverse PLP pointer.
\end{lemma}

\begin{proof}
  By categorizing the edges of the tree into primary or secondary ones,
  one can induce a decomposition of the tree to
  a disjoint set of paths (see also Figure~\ref{fig:leaf_pointer_decomposition}).
  Here, each path $\mathcal{P}$ in the decomposition
  begins with a secondary node $x$, continues with a series of primary nodes (possibly empty), and ends with a leaf $u$.
  Thus, the leaf $u$ is pointed by the primary leaf pointer from the secondary node $x$ that is the beginning node of the path $\mathcal{P}$.
  It is clear that the leaf $u$ is pointed only by this secondary node $x$
  that belongs to the same path $\mathcal{P}$. \qed
\end{proof}

When the suffix tree for $W_{i-1} = T[\ell_{i-1}..r_{i-1}]$ is updated to
the suffix tree for $W_i = T[\ell_i..r_i]$ at the $i$-th iteration,
the structure of the suffix tree can change.
This structural change requires us to maintain primary leaf pointers efficiently,
as some leaves and internal nodes can be removed and/or added.
The following property is known:

\begin{lemma}[Larsson~\cite{Larsson1996}] \label{lem:total_updates}
  For any string $T$ of length $n$, the total number of nodes removed and/or added is $O(n)$ regardless of the (fixed) size of the sliding window.
\end{lemma}

Below, we describe our algorithm which
maintains primary leaf pointers in $O(1)$ worst-case time
per leaf inserted into or deleted from the sliding suffix tree.
This together with Lemma~\ref{lem:total_updates}
will give us an $O(n)$ total time updates of the sliding suffix tree.

The key to our efficient algorithm is Lemma~\ref{lem:onePLP}
which states that the primary leaf pointer function $\plp{\cdot}$ is an injection.
Namely, only a constant number of primary leaf pointers need to be
updated per leaf insertion/deletion.
We will also assume that the inverse pointers of PLPs are maintained, so that
when $\plp{x} = u$ we can also access $x$ from $u$ in constant time.
Whenever we create, delete, or redirect a PLP, we also update its inverse pointer accordingly.

In what follows, the \emph{primary path} of a PLP from a secondary node $x$ denotes the path consisting of the primary edge(s) from $x$ to the leaf destination of $\plp{x}$.
Similarly, the primary path of any node $x$ refers to the path of primary edges from $x$ to a leaf.

\subsection{Handling node insertions on nodes}\label{sec:case1}
Our procedure for inserting leaf $u$ as a child of an existing internal node $w$ is as follows:

\begin{description}
  \item[Case 1-1:] When $w$ has no child (Figure~\ref{fig:Case1-1}).
    \begin{description}
      \item[Procedure:] Make $u$ primary and set $\plp{w}=u$. 
      \item[Correctness:] 
        This case only happens when $w$ is the root
        because inserting a new leaf under an existing leaf does not occur in suffix trees of non-empty strings.
        The correctness is trivial.
    \end{description}

  \item[Case 1-2:] When $w$ has a child (Figure~\ref{fig:Case1-2}).
    \begin{description}
      \item[Procedure:] Make $u$ secondary and initialize $\plp{u}=u$.
      \item[Correctness:]
        In this case, no existing primary path changes, and therefore no existing PLP needs to be updated.
        A new secondary node $u$ is created, thus $\plp{u}$ needs to be created. Since $u$ is a leaf, $u$ itself is the correct destination of $\plp{u}$.
    \end{description}
\end{description}

\begin{figure}[h]
  \centering
  \begin{minipage}{.5\textwidth}
    \centering
    \includegraphics[scale=0.45]{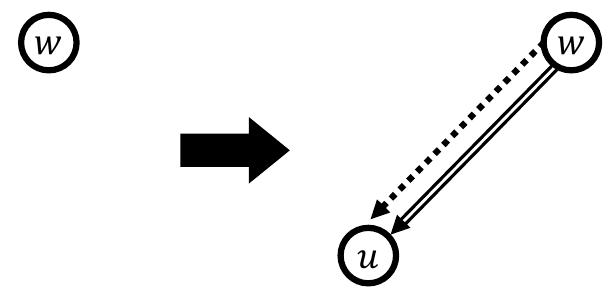}
    \caption{Case 1-1.}
    \label{fig:Case1-1}
    \end{minipage}\begin{minipage}{.5\textwidth}
    \centering
    \includegraphics[scale=0.45]{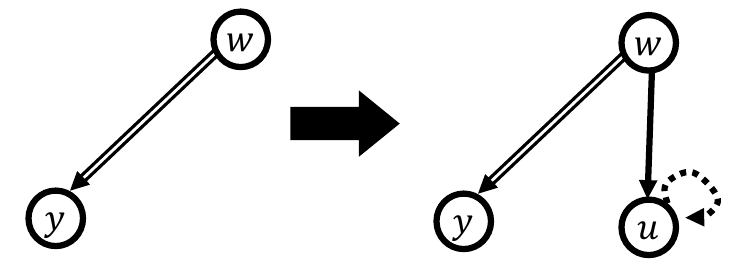}
    \caption{Case 1-2.}
    \label{fig:Case1-2}
  \end{minipage}
\end{figure}

\subsection{Handling node insertions on edges}\label{sec:case2}
Here, we describe our procedure for inserting a leaf $u$, such that an edge $x\rightarrow y$ is split into $x\rightarrow w$ and $w\rightarrow y$, where $w$ is a newly created branch node and is the parent of $u$.

\begin{description}
  \item[Case 2-1:] When $y$ is primary (Figure~\ref{fig:Case2-1}).
    \begin{description}
      \item[Procedure:] Make $w$ primary, $u$ secondary, and initialize $\plp{u}=u$. 
      \item[Correctness:]     In this case, the primary path that includes $x\rightarrow y$ is split.
        Thus, the primary path that contained $x\rightarrow y$ has changed.
        However, its destination clearly remains the same, namely the leaf at the end of the primary path of $y$.
        Thus, no change for existing PLPs is needed.
        Here, again a secondary leaf node $u$ is created, and thus its PLP needs to be created to point to itself.

    \end{description}

  \item[Case 2-2:] When $y$ is secondary (Figure~\ref{fig:Case2-2}).
    \begin{description}
      \item[Procedure:] Make $w$ secondary, $u$ primary, and initialize $\plp{w}=u$.
      \item[Correctness:]     In this case, the edge being split $x\rightarrow y$ is secondary, thus no existing PLP changed its primary path.
        Additionally, a new secondary node $w$ is created, so $\plp{w}$ needs to be created and initialized.
        Since $y$ is made to remain secondary, the primary child of $w$ has to be the newly created $u$.
        Since $u$ is a primary leaf child of $w$, it is the correct destination of $\plp{w}$.

    \end{description}
\end{description}

\begin{figure}[h]
  \centering
  \begin{minipage}{.5\textwidth}
    \centering
    \includegraphics[scale=0.45]{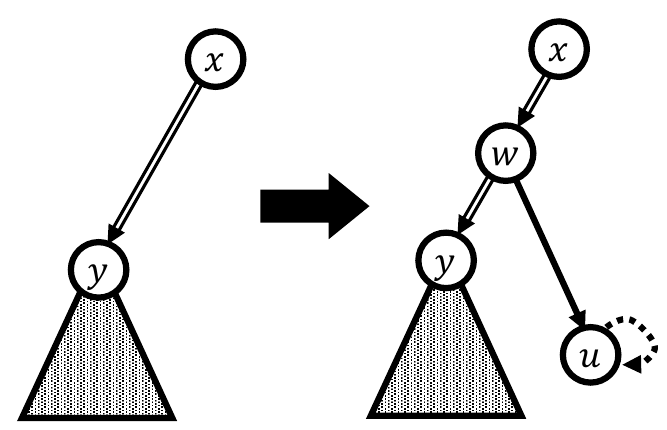}
    \caption{Case 2-1.}
    \label{fig:Case2-1}
    \end{minipage}\begin{minipage}{.5\textwidth}
    \centering
    \includegraphics[scale=0.45]{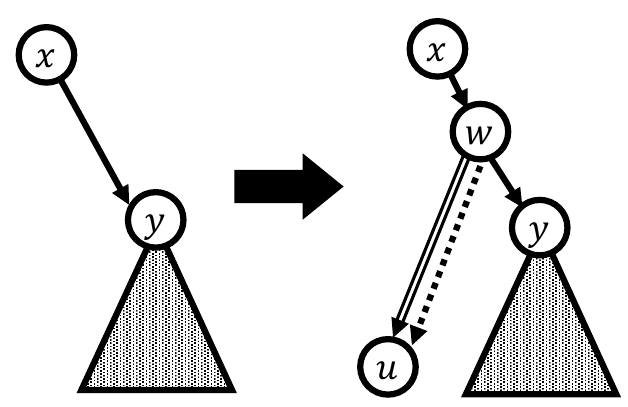}
    \caption{Case 2-2.}
    \label{fig:Case2-2}
  \end{minipage}
\end{figure}

\subsection{Handling leaf deletions under root parent node}\label{sec:case3}
Here, we describe our procedure for deleting a leaf $u$, when its parent node $w$ is the root.
Note that no merging occurs in Cases 3-1--3-3 since $w$ has no parent.

\begin{description}
  \item[Case 3-1:] When $u$ is primary and $w$ has no other children (Figure~\ref{fig:Case3-1}).
    \begin{description}
      \item[Procedure:] Set $\plp{w}=w$.
      \item[Correctness:] In this case, the tree consists only of the root after the update, thus the only PLP that exists is from the root $w$ pointing to itself.
    \end{description}

  \item[Case 3-2:] When $u$ is primary and $w$ has at least one other child (Figure~\ref{fig:Case3-2}).
    \begin{description}
      \item[Procedure:] Make $y$ primary and set $\plp{w}=v$, where $y$ is any secondary child of $w$ and $v$ is the leaf that $\plp{y}$ points to.
      \item[Correctness:] Here, the primary child $u$ of $w$ is deleted, but $w$ has at least one more child $y$ remaining.
        Thus, $y$ is changed into primary to maintain the invariant of child-having nodes having exactly one primary child.
        Consequently, $y$ no longer has a PLP after the update.
        Also, the primary path from $w$ changes into $w\rightarrow y\rightarrow \cdots$.
        Since the subtree of $y$ remains unchanged, clearly the destination of this primary path, which is also the correct new destination of $\plp{w}$,
        is the former destination of $\plp{y}$ before the update. 
    \end{description}   

  \item[Case 3-3:] When $u$ is secondary (Figure~\ref{fig:Case3-3}).
    \begin{description}
      \item[Procedure:] No further updates.
      \item[Correctness:] Since the node and edge deleted are secondary, and since no merging occurs, all existing primary paths remain the same, necessitating no change in existing PLPs that remain (i.e., undeleted).
    \end{description}
\end{description}

\begin{figure}[h]
  \begin{minipage}{.25\textwidth}
    \centering
    \includegraphics[scale=0.4]{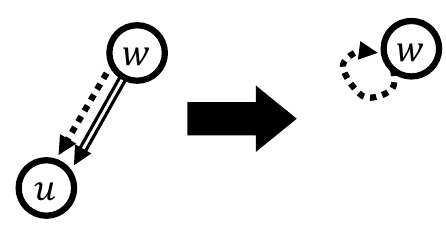}
    \caption{Case 3-1.}
    \label{fig:Case3-1}
  \end{minipage}
  \hfill
  \begin{minipage}{.4\textwidth}
    \centering
    \includegraphics[scale=0.4]{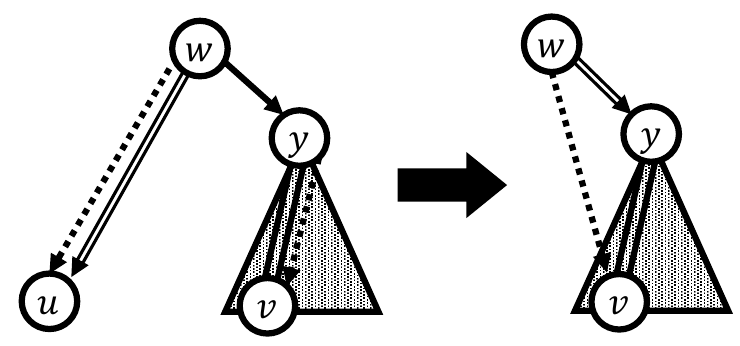}
    \caption{Case 3-2.}
    \label{fig:Case3-2}
  \end{minipage}
  \hfill
  \begin{minipage}{.3\textwidth}
    \centering
    \includegraphics[scale=0.4]{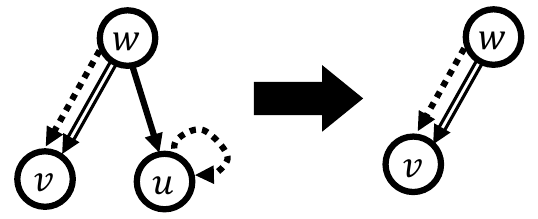}
    \caption{Case 3-3.}
    \label{fig:Case3-3}
  \end{minipage}
\end{figure}

\subsection{Handling primary leaf deletions under non-root parent node}\label{sec:case4}
Here, we describe our procedure for deleting a primary leaf $u$, when its parent node $w$ is not the root.
Let $x$ be the parent of $w$, $y$
any secondary child of $w$, $v$ the leaf $\plp{y}$ points to,
and $z$ the node such that $\plp{z}=u$.
Since $w$ is not the root, it has at least two children.

\begin{description}
  \item[Case 4-1:] When $w$ has more than two children (Figure~\ref{fig:Case4-1}).
    \begin{description}
      \item[Procedure:] Make $y$ primary, and set $\plp{z}=v$.
      \item[Correctness:] In this case, $z$ is $w$ or its ancestor, and $w$ does not get deleted. Here, the deletion of the primary node $u$ and its incoming primary edge necessitates the appointment of another child of $w$ to become primary,
        and also induces a change on the primary path of the PLP that pointed to $u$, namely that from $z$.
        Similarly to Case 3-2, $y$ is chosen as the new primary child of $w$ and thus the primary path from $z$ now goes into $y$,
        and due to the unchanged subtree of $y$, $v$ is the correct new destination of $\plp{z}$.
    \end{description}

  \item[Case 4-2:] When $w$ has two children and is primary (Figure~\ref{fig:Case4-2}).
    \begin{description}  
      \item[Procedure:] Make $y$ the primary child of $x$, and set $\plp{z}=v$.
      \item[Correctness:] In this case, $z$ is $x$ or its ancestor.
        Since $w$, which was primary, becomes non-branching and gets deleted due to merging, its parent $x$ needs to appoint a new primary child, for which $y$ is chosen.
        Thus, the primary path from $z$ now goes into $y$, and again we know the primary path destination from $y$ is $v$, the former destination of $\plp{y}$
        before the update, and therefore $v$ is the correct destination of $\plp{z}$.

    \end{description}

  \item[Case 4-3:] When $w$ has two children and is secondary (Figure~\ref{fig:Case4-3}).
    \begin{description}
      \item[Procedure:] Make $y$ a secondary child of $x$.
      \item[Correctness:] In this case, $z=w$, and $w$ gets deleted due to merging.
        Since $w$ is secondary, the only affected primary path is the deleted path from $w$ to $u$; all primary paths corresponding to remaining PLPs stay unchanged.
    \end{description}
\end{description}

\begin{figure}[h]
  \centering
  \includegraphics[scale=0.45]{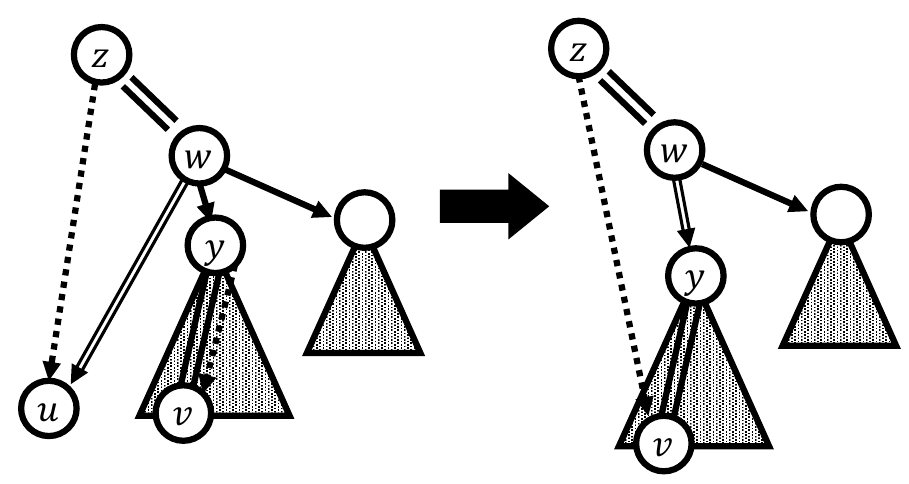}
  \caption{Case 4-1.}
  \label{fig:Case4-1}
\end{figure}

\begin{figure}[h]
  \begin{minipage}{.5\textwidth}
    \centering
    \includegraphics[scale=0.45]{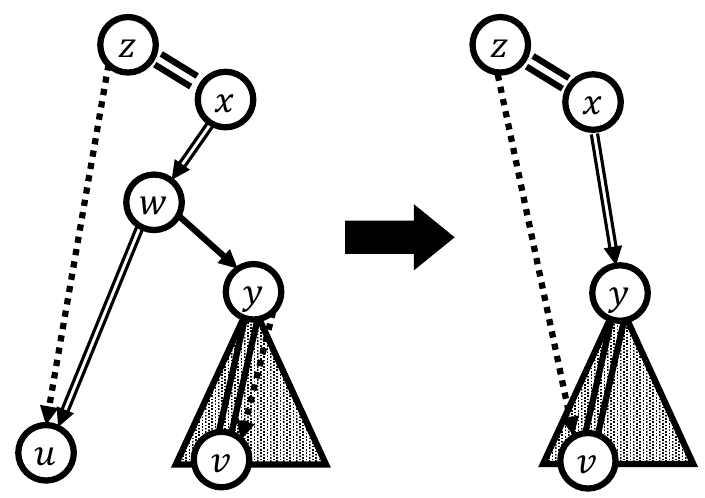}
    \caption{Case 4-2.}
    \label{fig:Case4-2}
  \end{minipage}
  \hfill
  \begin{minipage}{.5\textwidth}
    \centering
    \includegraphics[scale=0.45]{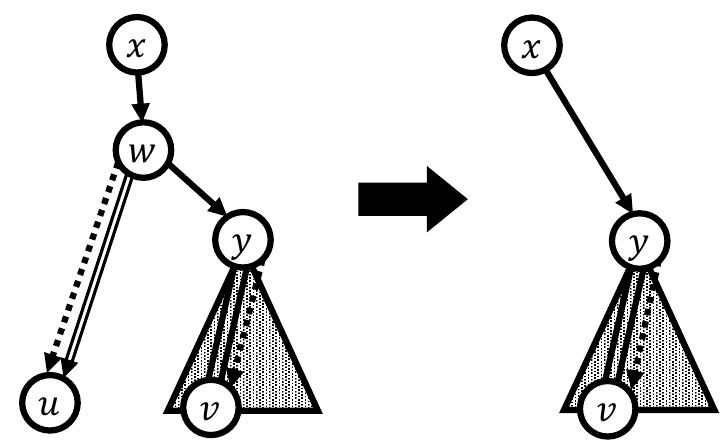}
    \caption{Case 4-3.}
    \label{fig:Case4-3}
  \end{minipage}
\end{figure}

\subsection{Handling secondary leaf deletions under non-root parent node}\label{sec:case5}
Here, we describe our procedure for deleting a secondary leaf $u$ and its parent node $w$ is not the root.
Let $x$ be the parent of $w$ and
$y$ the primary child of $w$.
Since $w$ is not the root, it has at least two children.

\begin{description}
  \item[Case 5-1:] When $w$ has more than two children (Figure~\ref{fig:Case5-1}).
    \begin{description}
      \item[Procedure:] No further updates.
      \item[Correctness:] In this case, $w$ does not get deleted.  Since only a secondary node is deleted and no merging occurs, no existing primary paths change.
    \end{description}

  \item[Case 5-2:] When $w$ has two children and is primary (Figure~\ref{fig:Case5-2}).
    \begin{description}
      \item[Procedure:] Make $y$ the primary child of $x$.
      \item[Correctness:] In this case, the deletion of $u$ causes $w$ to be deleted due to merging, and thus the primary path that included $x\rightarrow w \rightarrow y$ changes.
        However, the destination remains the same, namely the destination of the primary path from $y$. Thus, no existing PLPs change.
        Additionally, $x$ remains having exactly one primary child, by letting $y$ remain primary.

    \end{description}   

  \item[Case 5-3:] When $w$ has two children and is secondary (Figure~\ref{fig:Case5-3}).
    \begin{description}
      \item[Procedure:]
        Make $y$ a secondary child of $x$ and set $\plp{y}=v$,
        where $v$ is the leaf pointed to by $\plp{w}$.
      \item[Correctness:] Here, the deletion of $u$ causes $w$ to be deleted due to merging.
        Since $y$ is made secondary, the primary child of $x$ remains unchanged, but $\plp{y}$ needs to be initialized.
        Since $y$ was formerly primary, the former destination of $\plp{w}$ is in the unchanged subtree of $y$, namely leaf $v$,
        and therefore $v$ is the correct destination of the new $\plp{y}$.

    \end{description}
\end{description}

\begin{figure}[h]
  \centering
  \includegraphics[scale=0.45]{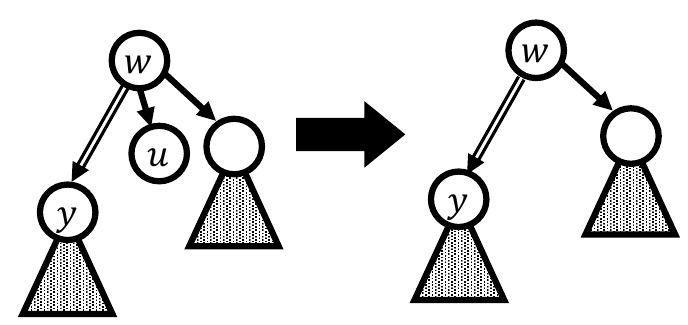}
  \caption{Case 5-1.}
  \label{fig:Case5-1}
\end{figure}

\begin{figure}[h]
  \centering
  \begin{minipage}{.5\textwidth}
    \centering
    \includegraphics[scale=0.45]{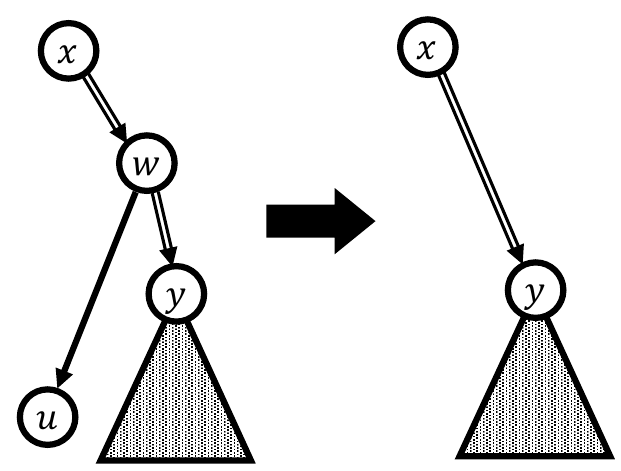}
    \caption{Case 5-2.}
    \label{fig:Case5-2}        
    \end{minipage}\begin{minipage}{.5\textwidth}
    \centering
    \includegraphics[scale=0.45]{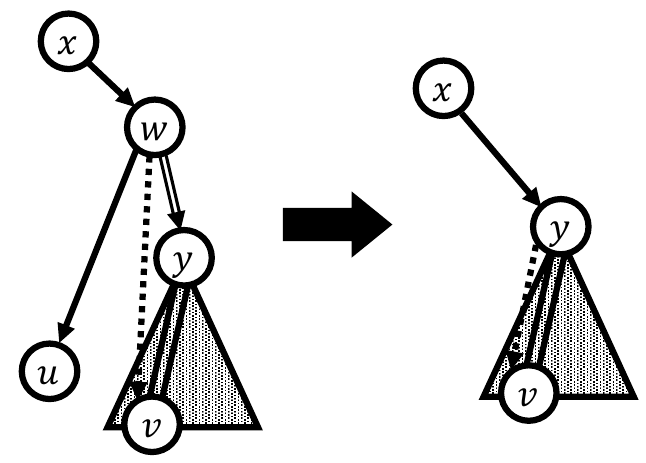}
    \caption{Case 5-3.}
    \label{fig:Case5-3}
  \end{minipage}
\end{figure}

To illustrate the whole picture of our proposed method, we show in Figures~\ref{fig:cases_insert} and~\ref{fig:cases_delete} the decision trees of our leaf-pointer maintenance algorithm
for the cases of right-end insertions and left-end deletions, respectively.

\begin{figure}[tbh]
  \centering
  \includegraphics[scale=0.75]{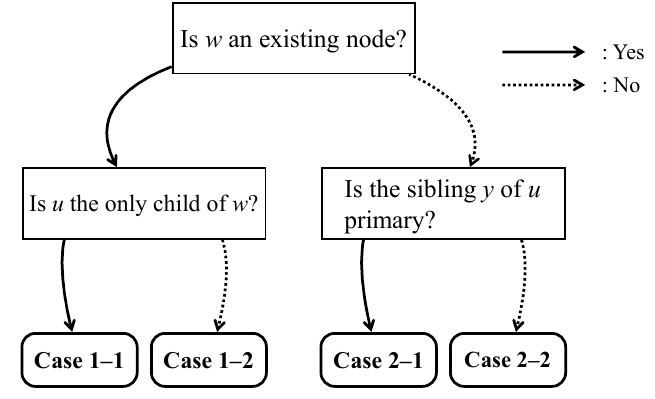}
  \caption{
    The decision tree representing the four cases of Sections~\ref{sec:case1} and \ref{sec:case2},
    when a new node is inserted into the suffix tree
    where $u$ is the new leaf
    and $w$ is the parent of $u$.
  }\label{fig:cases_insert}
\end{figure}

\begin{figure}[tbh]
  \centering
  \includegraphics[width=\linewidth]{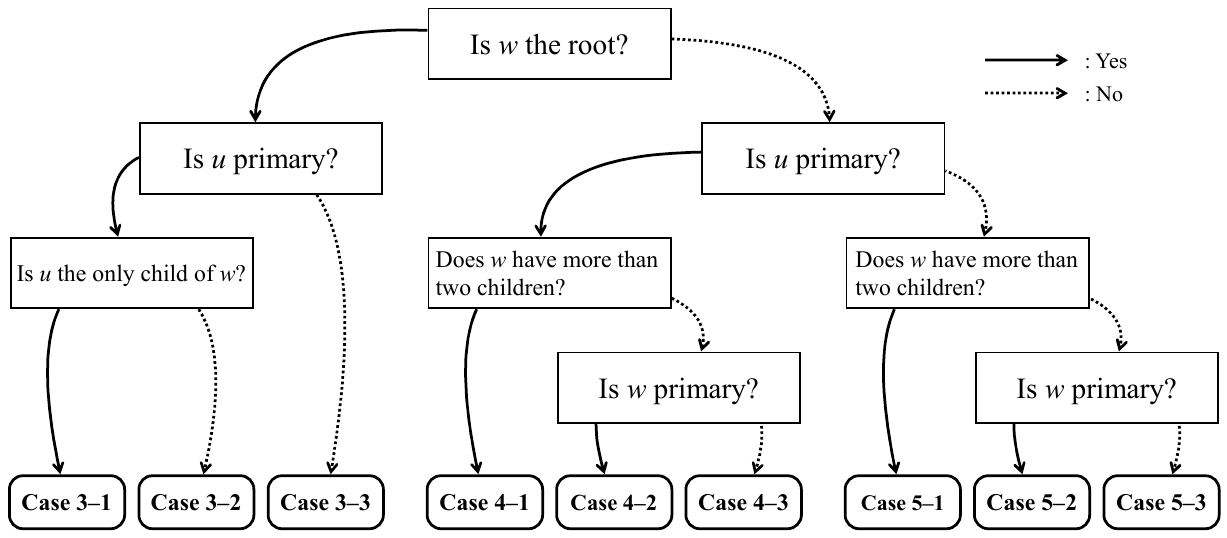}
  \caption{
    The decision tree representing the nine cases of Sections~\ref{sec:case3}, \ref{sec:case4}, and \ref{sec:case5},
    when an old leaf is deleted from the suffix tree
    where $u$ is the leaf to be deleted
    and $w$ is the parent of $u$.
  }\label{fig:cases_delete}
\end{figure}

We conclude this section with the following theorem:

\begin{theorem}
  Our algorithm described above correctly maintains
  leaf pointers on the sliding suffix tree for $W_i = T[\ell_i..r_i]$
  in $O(1)$ time per node insertion/deletion, using $O(d)$ working space.
\end{theorem}

\begin{proof}
  The described steps to maintain the invariants cover all possible cases of insertion and deletion in a sliding suffix tree, and each step can clearly be handled in $O(1)$ worst-case time.
  This can also be observed from our decision trees in Figures~\ref{fig:cases_insert} and~\ref{fig:cases_delete}.
  The correctness of all the cases has already been shown.
  This allows us to answer any leaf pointer query in $O(1)$ worst-case time.
  The working space is clearly $O(d)$. \qed
\end{proof}

We remark that our leaf-pointer maintenance algorithm works
in $O(1)$ worst-case time even when the window size $d$ is not fixed,
and the working space is linear in the current window size.

\section{$\Theta(d)$-time worst-case lower bound for credit-based method}
\label{sec:loweround}
In this section, we present an instance that requires $\Theta(d)$ worst-case time
per single leaf insertion or deletion
with the credit-based method~\cite{Larsson1996,Larsson1999},
where $d$ is the window size.

Here, we briefly recall the credit-based method
by Larsson~\cite{Larsson1996,Larsson1999}, with a correction by Senft~\cite{Senft2005} for leaf deletions.
When a leaf is created, then it issues a \emph{credit} to its parent.
Each internal node $u$ maintains a credit counter $\credit(u) \in \{0,1\}$.
We initially set $\credit(u) = 0$ when an internal node $u$ is created.

When an internal node $u$ receives a credit from a child $v$,
then there are two cases:
\begin{enumerate}
  \item If $\credit(u) = 0$, then we set $\credit(u) = 1$.
  \item Otherwise (if $\credit(u) = 1$), then we first spend the existing credit in $u$ by setting $\credit(u) = 0$, and update the label of the edge leading to $u$.
    Second, we pass, to the parent of $u$, the credit that came from the child $v$.
\end{enumerate}

When an internal node $v$ is deleted and two adjacent edges $u \to v$ and $v \to w$ are merged into a single edge $u \to w$, then two credits are issued from $v$ independently of the value of $\credit(v)$.
Then, one credit is used to update the label of the edge $u \to w$,
and the other credit is passed to the parent $u$ of the deleted node $v$.

While we omit the details in this paper, each edge label update can be done in constant time,
and the above procedure keeps all edge labels valid~(see~\cite{Larsson1996,Larsson1999,Senft2005} for details).
At most $n$ leaves are created and deleted
for all sliding windows $W_i$ with $1 \leq i \leq n-d+1$ over a string of length $n$.
Since each leaf issues at most three credits (one for its insertion and two for its deletion), the total number of credits issued is at most $3n$.
Since each edge label update spends a credit,
the time for updating an edge label is amortized $O(1)$ per leaf insertion or deletion.

In the next theorem,
we show that this credit-based method requires $\Theta(d)$ time
per leaf insertion or deletion in the worst case.

\begin{theorem}
  There exists a string for which 
  Larsson's credit-based sliding suffix tree algorithm~\cite{Larsson1996,Larsson1999} requires $\Theta(d)$ edge-label updates for some single leaf insertion or deletion, where $d$ is the window size.
\end{theorem}

\begin{figure}[tbh]
  \centering
  \includegraphics[scale=0.4]{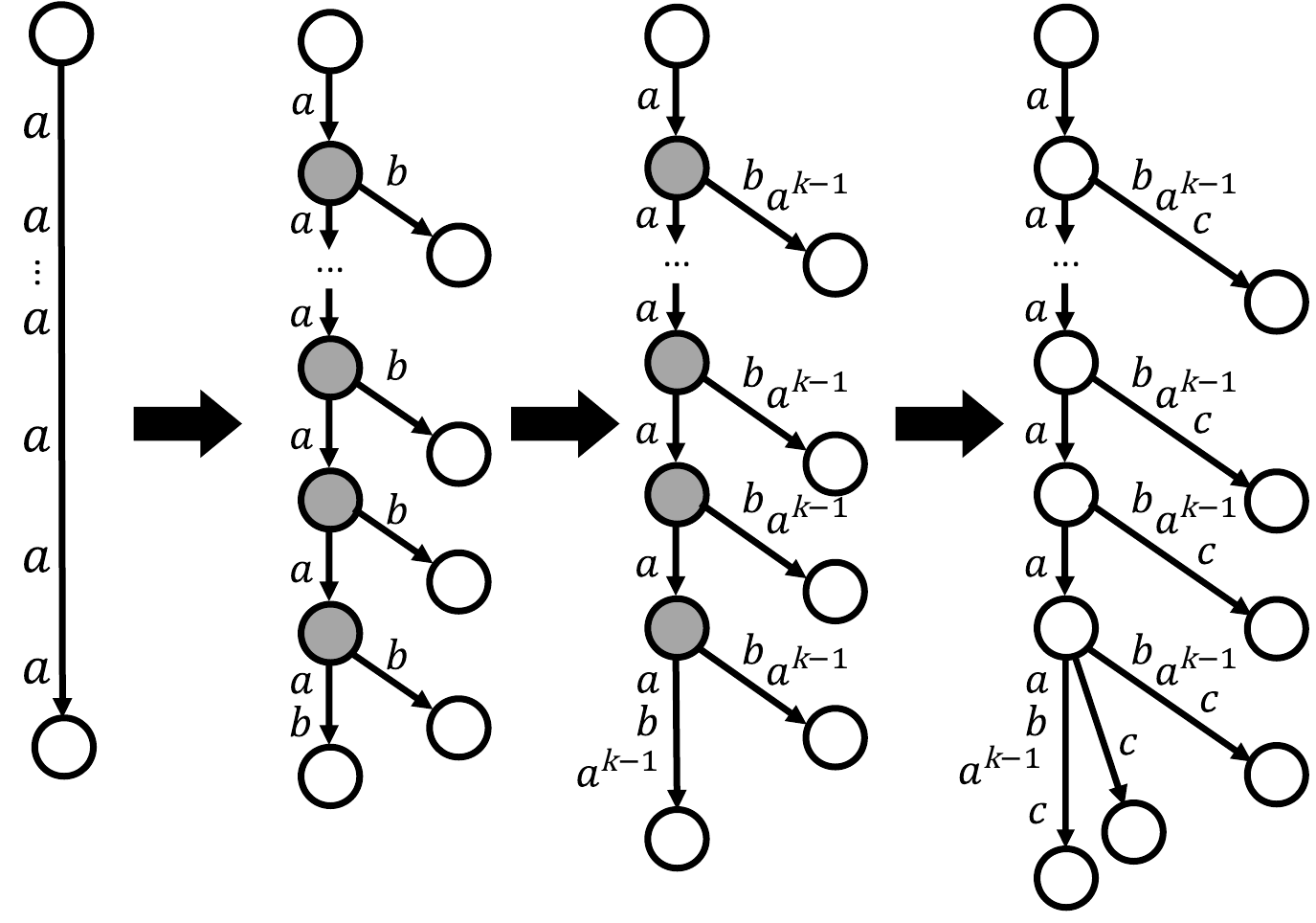}
  \caption{Illustration of $\Theta(d)$-time update of credit-based methods per single leaf insertion. The gray and white nodes represent internal nodes with credit value 1 and 0, respectively. Note that the subtrees starting with characters other than $a$ are omitted for clarity.}
  \label{fig:WorstCaseInsertion}
\end{figure}

\begin{proof}
  The following credit configuration is reachable by the standard Ukkonen-style insertion phase used in the sliding suffix tree algorithm.
  Here we measure the worst-case cost of an individual leaf-insertion suboperation occurring during one window shift; the deletion suboperation in the same shift is not charged to this leaf insertion.

  For the case of leaf insertion, consider a string that has $W = a^kba^{k-1}$ as its prefix
  of length $2k = d$.
  We then append a new character $c$ to the right end of $W$.
  See also Figure~\ref{fig:WorstCaseInsertion} for illustrations.
  Since our string begins with $W = a^kba^{k-1}$,
  we first build $\STree(W)$ in an online manner, from left to right,
  as is done by Ukkonen's construction~\cite{Ukkonen}.
  Below, we explain how the sliding suffix tree is built and updated.
  \begin{itemize}
    \item Build $\STree(a^k)$ that consists only of a single path (see also the first step of Figure~\ref{fig:WorstCaseInsertion}).

    \item Then, we update it to $\STree(a^kb)$,
      in which the new leaves $a^hb$ and their parents $a^h$
      are inserted in descending order for $h = k-1, \ldots, 1$ (see also the second step of Figure~\ref{fig:WorstCaseInsertion}).
      When the internal node $a^h$ is inserted, we initially set $\credit(a^h) = 0$.
      As soon as the leaf $a^hb$ is inserted as a new child of $a^h$,
      then $a^h$ receives a credit from this new leaf and
      we set $\credit(a^h) = 1$.
      Thus, every node $a^h$ for $1 \leq h \leq k-1$ obtains $\credit(a^h) = 1$
      on $\STree(a^kb)$.

    \item Next, we update $\STree(a^kb)$ to $\STree(a^kba^{k-1})$ (see also the third step of Figure~\ref{fig:WorstCaseInsertion}).
      No new leaves are inserted during this update.

    \item Finally, we update $\STree(a^kba^{k-1})$ to $\STree(a^kba^{k-1}c)$ (see also the fourth step of Figure~\ref{fig:WorstCaseInsertion}).
      The first new leaf inserted is $a^{k-1}c$, which is a new child of $a^{k-1}$.
      Since $\credit(a^{k-1}) = 1$,
      we reset it as $\credit(a^{k-1}) = 0$ and pass the credit from the child $a^{k-1}c$ to the parent $a^{k-2}$.
      This propagates to all $a^h$ in decreasing order of $h = k-1, \ldots, 1$.
  \end{itemize}
  Thus, the credit-based method
  requires $\Theta(d)$ edge-label updates
  for a single leaf insertion in the worst case.

  For the case of leaf deletions, we consider a string $W' = a^{d}b$ of length $d+1$.
  See also Figure~\ref{fig:WorstCaseDeletion} for illustrations.
  We have $\credit(a^h) = 1$ for all internal nodes $a^h$ with
  $1 \leq h \leq d-1$ in $\STree(a^{d}b)$.
  Then we delete the leftmost character $a$ from $W'$
  and update $\STree(a^{d}b)$ to $\STree(a^{d-1}b)$.
  This deletes the leaf corresponding to the longest suffix $a^{d}b$ and its parent $a^{d-1}$ that becomes non-branching.
  After the deletion of the internal node $a^{d-1}$,
  a credit is sent to its parent $a^{d-2}$ such that $\credit(a^{d-2}) = 1$.
  We reset $\credit(a^{d-2}) = 0$ and pass the credit to its parent $a^{d-3}$.
  This propagates to all $a^h$ for decreasing $h = d-1, \ldots, 1$.
  Thus, the credit-based method requires $\Theta(d)$ edge label updates for a single leaf deletion in the worst case. \qed
\end{proof}

\begin{figure}[t!]
  \centering
  \includegraphics[scale=0.4]{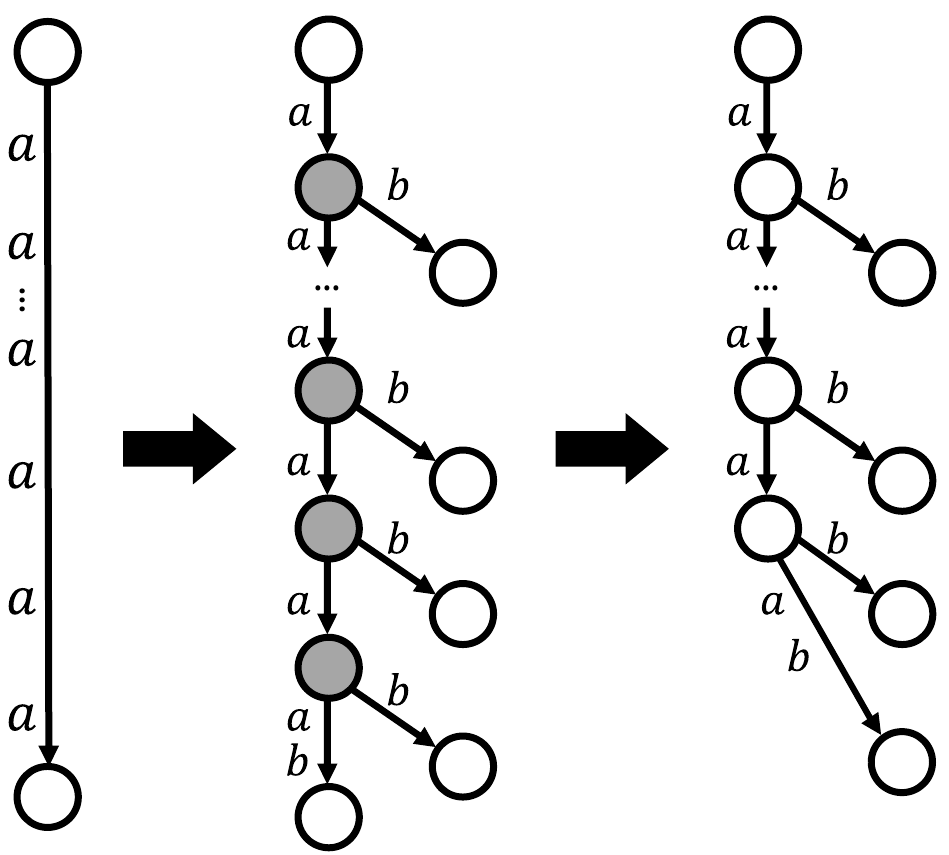}
  \caption{Illustration of $\Theta(d)$-time updates of credit-based methods per single leaf deletion.}
  \label{fig:WorstCaseDeletion}
\end{figure}
\section*{Acknowledgements}
The authors thank the anonymous referees of earlier versions of this paper
for their helpful comments and suggestions for improvement.
This work was supported by JSPS KAKENHI Grant Numbers
JP23K24808, JP23K18466 (SI),
JP24K02899 (HB), and
JP24K20734 (TM).

\end{document}